\documentclass{llncs}
\usepackage{makeidx}
\usepackage{stmaryrd}
\usepackage{amsmath}
\usepackage{framed}
\usepackage{amsfonts}
\usepackage{graphicx}
\usepackage{graphicx}
\usepackage[
latex={-interaction=nonstopmode},
crop=off,runs=2
]{auto-pst-pdf} 
\usepackage{psfrag}
\newtheorem{thm}{Theorem}[section]

\newtheorem{lem}[thm]{Lemma}

\newtheorem{prop}[thm]{Proposition}

\numberwithin{equation}{section}

\newcommand{\eps}{\varepsilon}

\newcommand{\A}{\mathcal{A}}

\def\span{\operatorname{span}}
\def\sign{\operatorname{sign}}

\def\R{\mathbb R}
\def\eps{\varepsilon}

\def\vol{{\operatorname{vol}}}

\newcommand{\row}[2]{{#1}(#2,:)}
\newcommand{\col}[2]{{#1}(:,#2)}
\renewcommand{\cos}{\operatorname{cos}}
\renewcommand{\sin}{\operatorname{sin}}

\begin{document}

\frontmatter
\pagestyle{headings}
\mainmatter
\title{Tighter Fourier Transform Complexity Tradeoffs }
\author{Nir Ailon  \inst{1}}
\institute{Technion Israel Institute of Technology, Haifa, Israel\\
\email{nailon@cs.technion.ac.il}
}
\maketitle

\def\sign{\operatorname{sgn}}
\def\dim{n}
\def\KL{{\operatorname{KL}}}
\def\C{\mathbb C}
\def\R{\mathbb R}
\def\xover{x_{\operatorname{over}}}
\def\xunder{x_{\operatorname{under}}}
\def\P{\mathcal P}
\def\Q{\mathcal Q}
\def\trace{\operatorname{tr}}
\def\diag{\operatorname{diag}}
\def\rank{\operatorname{rank}}
\def\F{{\mathcal F}}
\def\Id{\operatorname{Id}}
\def\f{{\hat f}}
\def\Z{{\mathbb Z}}
\def\X{{\cal X}}
\def\B{{\cal B}}
\def\Ellips{{\cal E}}
\def\S{{\cal S}}
\def\err{{\operatorname{err}}}
\def\MLAE{{\operatorname{MLAE}}}
\def\N{{\mathcal N}}
\def\tr{\operatorname{tr}}
\def\Ellip{{\cal E}}
\def\I{{\cal I}}

\begin{abstract}
The Fourier Transform is one of the most important linear transformations used in science and engineering.  Cooley and Tukey's Fast Fourier Transform (FFT) from 1964 is a method for computing this transformation in time $O(n\log n)$.   Achieving
a matching lower bound in a reasonable computational model is one of the most important open problems in theoretical
computer science.

In 2014, improving on his previous work, Ailon showed that  if an algorithm speeds up the FFT by a factor of $b=b(n)\geq 1$,  then it must rely on computing, as an intermediate ``bottleneck'' step,  a linear mapping of the input  with condition number $\Omega(b(n))$.  Our main result shows that a factor $b$ speedup implies existence of not just one but $\Omega(n)$  $b$-ill conditioned bottlenecks occurring at $\Omega(n)$ different steps, each causing  information from independent (orthogonal) components of the input to either overflow or underflow.   This provides further evidence that beating FFT is hard.  Our result also gives the first quantitative tradeoff between computation speed and information loss in Fourier computation on fixed word size architectures. The main technical result is an entropy analysis of the Fourier transform under transformations of low trace, which is  interesting in its own right.
\end{abstract}

\section{Introduction}
The (discrete) normalized Fourier transform  (DFT) is a complex  mapping sending  input $x\in \C^n$ to $Fx\in \C^n$, where $F$ is a unitary  matrix defined by
\begin{equation}\label{dft} F(k,\ell) = n^{-1/2}e^{-i2\pi k\ell/n}\ .
\end{equation}
The Walsh-Hadamard transform is a real orthogonal mapping in $\R^n$ (for $n$ an integer power of $2$) sending
an input $x$ to $Fx$, where $$F(k,\ell)=\frac 1 {\sqrt n}  (-1)^{\langle [k-1], [\ell-1]\rangle}\ ,$$ with $\langle\cdot,\cdot\rangle$ is dot-product, and $[p]$ denotes (here only) the bit
representation of the integer $p\in \{0,\dots, n-1\}$ as a vector of $\log_2 n$ bits. 
Both transformations are special (and most important) cases of  abstract Fourier transforms defined with respect to corresponding Abelian groups.
The Fast Fourier Transform (FFT) of Cooley and Tukey \cite{CooleyT64} is a method for computing the DFT of  $x\in \C^n$
in time $O(n\log n)$.  The fast Walsh-Hadamard transform  computes the Walsh-Hadamard
transform in time $O(n\log n)$.  Both fast transformations perform a sequence of rotations on pairs of coordinates, and
are hence special cases of so-called linear algorithms, as defined in
\cite{Morgenstern:1973:NLB:321752.321761}.  

The DFT is instrumental as a subroutine in fast polynomial multiplication \cite{CLRS} (chapter 30), 
fast integer multiplication \cite{DeKSS13,Furer:2007:FIM:1250790.1250800}, cross-correlation and auto-correlation detection in 
images and time-series (via convolution) and, as a more recent example, convolution networks for deep learning \cite{MathieuHLC14}.  Both DFT and Walsh-Hadamard are useful for
fast Johnson-Lindenstrauss transform for dimensionality reduction \cite{DBLP:journals/siamcomp/AilonC09,DBLP:journals/dcg/AilonL09,DBLP:journals/talg/AilonL13,DBLP:journals/siamma/KrahmerW11} and the
related restricted
isometry property (RIP) matrix construction \cite{DBLP:journals/jacm/RudelsonV07,DBLP:journals/corr/abs-1301-0878,DBLP:journals/siamma/KrahmerW11}).
It is beyond the scope of this work to survey all uses of Fourier transforms in both theory of algorithms and in complexity.
For the sake of simplicity the reader is encouraged to assume that $F$ is the Walsh-Hadamard transform, and that by the acronym ``FFT'' we refer to the fast Walsh-Hadamard transform.  The modifications required for the DFT (rather, the real embedding thereof) require a slight modification to the potential function which 
we mention but do not elaborate on for simplicity.  Our results nevertheless apply also to DFT.


It is not known whether $\Omega(n\log n)$ operations are necessary, and this problem is one of the most
important open problems in theoretical computer science \cite{wiki}.
It is trivial  that a linear number of steps is necessary, because every input coordinate must be probed.
Papadimitriou  derives in \cite{Papadimitriou:1979:OFF:322108.322118} an $\Omega(n\log n)$ lower bound for DFT over finite fields using a notion of an information flow network.   It is not clear how to extend
that result to the Complex field.  There have also been attempts \cite{Winograd76} to reduce the constants hiding in the upper bound
of $O(n\log n)$, while also separately  counting the number of additions versus the number of multiplications (by constants).
In 1973, Morgenstern proved that if the moduli of the constants used in the computation are are bounded by $1$
then the number of steps required for computing the \emph{unnormalized} Fourier transform, defined by $n^{1/2}F$ in the linear algorithm model is at least 
$\frac 1 2 n\log_2 n$.  He used a potential function related to matrix determinant, which makes the technique inapplicable for
deriving lower bounds for the (normalized) $F$.  
Morgenstern's result also happens to imply that the transformation $\sqrt n\Id$ ($\sqrt n$ times the identity) has the same
complexity as the Fourier transform, which is not a satisfying conclusion.
    Also note that stretching the input norm by a factor of $\sqrt n$ requires representing numbers of  $\omega(\log n)$
    bits, and  it cannot be simply assumed that a multiplication or an addition over such numbers can be done in $O(1)$ time.

  Ailon \cite{Ailon13}  studied  the complexity of the (normalized) Fourier transform in a computational model
allowing only  orthogonal  transformations acting on (and replacing in memory) two intermediates at each step.  
He showed that at least $\Omega(n\log n)$ steps were required.
The proof was done by defining a potential function on the matrices $M^{(t)}$ defined by composing the first $t$ gates.  The potential
function is simply the sum of Shannon entropy of the probability distributions defined by the squared modulus of elements in the
matrix rows.  (Due to orthogonality, each row, in fact, thus defines a probability distribution).
That result had two shortcomings:  (i) The algorithm was assumed not to be allowed to use extra memory in addition to 
the space used to hold the input.  In other words, the computation was done \emph{in place}. (ii) The result was sensitive to the normalization of $F$, and was not useful in deriving
any lower bound for $\gamma F$ for $\gamma \not \in \{\pm 1\}$.

In \cite{Ailon14}, Ailon took another step forward by showing a lower bound for computing
 \emph{any scaling} of the Fourier transform in a stronger model of computation which we call  \emph{uniformly well conditioned}.
At each step, the algorithm can perform a nonsingular linear transformation on at most two intermediates, as long as
the matrix $M^{(t)}$ defining the composition of the first $t$ steps must have condition number at most $\kappa$, for all $i$. 
We remind the reader that condition number of a matrixis defined as 
the ratio between its largest and smallest (nonzero) singular values.
Otherwise stated, the result implies that if an algorithm computes the Fourier transform in time $(n\log n)/b$ for some $b>1$, then some
$M^{(t)}$ must have condition number at least $\Omega(b)$.  This means that the computation output relies on
an ill  conditioned intermediate step.  
The result in \cite{Ailon14} made a qualitative claim about  compromise of numerical stability due to a ill condition.
\subsection{Our Contribution}
Here we establish (Theorem~\ref{thm:subspaces})  that a $b$-factor speedup of FFT
for $b=b(n)=\omega(1)$ either \emph{overflows} at $\Omega(n)$ different time steps due to $\Omega(n)$ pairwise orthogonal input directions,
or \emph{underflows} at $\Omega(n)$ different time steps, losing accuracy of order $\Omega(b)$ at $n$  orthogonal
input directions.
 Note that achieving this could not be simply done  by a more careful
analysis of \cite{Ailon14}, but rather requires an intricate analysis of the entropy of Fourier transform under
transformations of small trace.  This analysis  (Lemma~\ref{lem:Fafterproj}) is interesting in its own right.

\section{Computational Model and Notation}\label{sec:modelnotation}
We remind the reader of the computational model discussed in \cite{Ailon13,Ailon14}, which is a special case of
the linear computational model.  The machine state represents a vector in $\R^\ell$ for some $\ell\geq n$,
where it initially  equals the input $x\in \R^n$ (with possible padding by zeroes, in case $\ell>n$).  Each step (gate)
is either a \emph{rotation} or a \emph{constant}.   A rotation applies a $2$-by-$2$ rotation mapping on a pair of
machine state coordinates (rewriting the result of the mapping to the two coordinates).  We remind the reader
that a $2$-by-$2$ rotation mapping is written in matrix form as $\left (\begin{matrix} \cos \theta & \sin \theta \\ -\sin \theta & \cos\theta\end{matrix}\right )$
for some real (angle) $\theta$.  A constant gate multiplies a single machine state coordinate (rewriting the result) by
a nonzero constant.   In case the constant equals $-1$, we call it a reflection gate.

In case $\ell=n$ we say that we are in the in-place model.   Any nonsingular linear mapping over $\R^n$ can be decomposed into a sequence of rotation and constant gates in the in-place model, and hence our model is, in a sense, universal.  FFT  works in the in-place model, using rotations (and possibly reflections) only.  A restricted method for dealing
 with $\ell>n$ was developed in \cite{Ailon14}, and can be applied here too in a certain sense (see Section~\ref{sec:future} for a discussion).  We focus in this work on the in-place model only.

Since both rotations and constants apply a linear transformation on the machine state, their composition is a linear transformation.  If $\A_n$ is an in-place algorithm for computing a  linear mapping over $\R^n$, it
is convenient to write it as $\A_n = (M^{(0)}=\Id, M^{(1)}, \dots, M^{(m)})$ where $m$ is the number of steps (gates),
$M^{(t)}\in \R^{n\times n}$ is the mapping that satisfies that for input $x\in \R^n$ (the initial machine state),
$M^{(t)}x$ is the machine state after $t$ steps.  ($\Id$ is the identity matrix).  The matrix $M^{(m)}$ is the target
transformation, which will typically be  $F$ in our setting.  In fact, due to the scale invariance of the potential function we use, we could take $M^{(m)}$ to be any nonzero scaling
of $F$, but to reduce notation we simply assume a scaling of $1$.
For any $t\in[m]$,  if  the $t$'th gate is a rotation, then $M^{(t)}$ defers
from $M^{(t-1)}$ in at most two rows, and if the $t$'th gate is a constant, then $M^{(t)}$ defers from $M^{(t-1)}$
in at most one row.


\subsection{Numerical Architecture}
 The in-place model 
implicitly assumes representation of a vector in $\R^n$ in memory using $n$ words.
A typical computer   word represents
 a coordinate (with respect to some fixed orthogonal basis)  in the range $[-1,1]$ to within some accuracy $\eps=\Theta(1)$.\footnote{The range $[-1,1]$ is immaterial and can be replaced with any range of the form $[-a,a]$ for $a>0$.}
  For sake of simplicity, $\eps$ should be thought of as  $2^{-31}$ or $2^{-63}$ in modern
computers of $32$ or $64$ bit words, respectively.  

To explain the difficulties in speeding up FFT on computers of fixed precision  in the in-place model, we need to  understand
whether (and in what sense) standard FFT is at all suitable on such machines.  First, we must restrict the domain of inputs.  Clearly this domain cannot be $\R^n$, because computer words can only represent
coordinates in the range $[-1,1]$, by our convention.  
We consider input from an $n$-ball of radius $ \Theta(\sqrt{n})$, which we denote $\B(\Theta(\sqrt{n}))$. 
An $n$-ball is  invariant under orthogonal
transformations, and is hence a suitable domain.  
Encoding a single coordinate of such an input might require  $\omega(1)$ bits (an overflow).  However, using
well known tools from high dimensional geometry,
encoding a single coordinate of a \emph{typical} input chosen randomly from $\B(\Theta(\sqrt{n}))$ requires $O(1)$ bits, fitting inside a machine word.\footnote{By ``encoding'' here we simply
mean the base-$2$ representation of the integer $\lfloor x(i)/\eps \rfloor$.}
We hence take a statistical approach and define a state of overflow as trying to encode, in some fixed memory word (coordinate), a random number of $\omega(1)$ bits
in expectation, at a fixed
time step in the algorithm.  This definition allows us  to avoid dealing with  accommodation of integers requiring
super-constant bits and, in turn, with logical bit-operation complexity.  Although the definition might seem impractical at first, it allows
us to derive very interesting information vs computational speed tradeoffs.
(In the future work  Section~\ref{sec:future} we shall discuss allowing varying word sizes and its implications on complexity.)
By our  definition, standard FFT for input drawn uniformly from  $\B(\Theta(\sqrt n))$ does not overflow at
all, because any coordinate of the machine state at any step is tightly concentrated (in absolute value) around $\Theta(1)$.
It will be easier however to replace the uniform distribution from the ball with the multivariate Gaussian $\N(0,\Theta(n)\cdot\Id)$,
which is a good approximation of the former for large $n$.  With this assumption, any coordinate of the standard FFT machine
state at any  step follows the law $\N(0,\Theta(1)$).
By simple integration against the Gaussian measure, one can verify  that the expected number of bits required to encode
such a random variable (to within fixed accuracy $\eps$) is $\Theta(1)$, hence no overflow occurs.
 This input assumption together with the no-overflow
guarantee  {\bf will serve as our benchmark}.  

For further discussion on the numerical arhitecture and definition of \emph{overflow} we refer the reader, due to  lack of space, to Appendix~\ref{sec:discussion}.

\section{The Matrix Quasi-Entropy Function}\label{sec:notationentropy}

The set $\{1,\dots, q\}$ is denoted by $[q]$.
By $\R^{a\times b}$ we formally denote matrices of $a$ rows and $b$ columns. 
Matrix transpose is denoted by $(\cdot)^T$.  
 We use $(\cdot)^{-T}$ as shorthand for $((\cdot)^{-1})^T=((\cdot)^{T})^{-1}$.  
 If $A\in \R^{a\times b}$ is a matrix and $I$ is a subset of $[b]$, then (borrowing from Matlab syntax) $\col{A}{I}$ is the submatrix obtained
 by stacking the columns corresponding to the indices in $I$ side by side and $\row{A}{I}$ is the submatrix obtained
 by stacking the rows corresponding to the indices in $I$ one on top of the other.
  We shall also write, for $i\in [b]$, $\col{A}{i}$ and $\row{A}{i}$
 as shorthands for $\col{A}{\{i\}}$ and $\row{A}{\{i\}}$, respectively.  All logarithms are  base $2$.

We slightly abuse notation and extend the definition of the quasi-entropy function $\Phi(M)$ defined on nonsingular matrices $M$ from \cite{Ailon14}, as follows.  Given two matrix arguments $A,B\in \R^{a\times b}$ for some $a,b\geq 1$, $\Phi(A,B)$ is defined as
$$\sum_{i=1}^a\sum_{j=1}^b -A(i,j)B(i,j)\log |A(i,j)B(i,j)|\ .$$

This extends naturally to vectors, namely for $u,v\in \R^a$, $\Phi(u,v)$ is as above by viewing $\R^a$ as $\R^{a\times 1}$.
If $A,B\in \R^{a\times b}$ and $a,b$ are even, then we define the \emph{complex quasi-entropy} function $\Phi^\C(A,B)$ to be:
$$\sum_{i=1}^a\sum_{j=1}^{b/2} -(A(i,2j-1)B(i,2j-1) + A(i,2j)B(i,2j))\log |A(i,2j-1)B(i,2j-1) + A(i,2j)B(i,2j)|\ .$$
The function $\Phi^\C$ can be used for proving our results for the real representation of the complex DFT, which
we omit from this manuscript for simplicity.
The reason we need  this modification to $\Phi$ for DFT is explained in the proof of Lemma~\ref{lem:Fafterproj}, needed by  Theorem~\ref{thm:subspaces} below.  Elsewhere,
we will work (for convenience and brevity) only with $\Phi$.
Abusing notation, and following \cite{Ailon14}, we define for any nonsingular matrix $M$:
$\Phi(M) := \Phi\left (M,  M^{- T}\right )\ , \Phi^\C(M) := \Phi^\C\left(M, M^{-T}\right)$.
%
It is easy to see that $\Phi(F) = n\log n$ for the Walsh-Hadamard transform, because all matrix elements
are $\pm 1/\sqrt n$.  If $F$ is a real representation of the $(n/2)$-DFT, then clearly
$\Phi^\C(F) = n\log(n/2)$, because all matrix elements of the (complex representation of the) $(n/2)$-DFT are complex unit roots times
$(n/2)^{-1/2}$.

It will be also useful to consider a  generalization of the potential of a nonsingular matrix $M$,
by allowing linear operators acting on the rows of $M$ and $M^{-T}$, respectively.  More precisely, 
we will let $\Phi_{P,Q}(M)$  be shorthand for
$\Phi(MP, M^{-T}Q)$,
where $P,Q\in \R^{n\times a}$ are some mappings. (We will only be working with projection matrices $P,Q$ here).
Similarly, $\Phi_{P,Q}^\C(M, M^{-T}) := \Phi^C(MP, M^{-T}Q)$.

Finally, for any matrix $A \in \R^{n\times n}$, let $\sigma_1(A),\dots, \sigma_n(A)$ denote its singular values, where we use the
convention $\sigma_1(A) \geq \cdots \geq \sigma_n(A)$.  If $A$ is nonsingular, then the condition number $\kappa(A)$ is defined by  $\sigma_1(A)/\sigma_n(A)$.
For any matrix $A$, we let $\|A\|$ denote its spectral norm and $\|A\|_F$ its Frobenius norm.  If $x$ is
a vector, hence, $\|x\|=\|x\|_2=\|x\|_F$.  Let $\B$ denote the Euclidean unit ball in $\R^n$




\section{Genralized Ill Conditioned Bottleneck from Speedup}
We show that if an in-place algorithm $\A_n=(M^{(0)}=\Id,\dots, M^{(m)}=F)$ speeds up FFT by a factor of $b\geq 1$, then for some $t$ $M^{(t)}$  is ill conditioned (in a generalized sense, to be explained).  
This is a generalization of the main result in \cite{Ailon14}, with a simpler proof that we provide in Appendix~\ref{sec:proof:thm:main}
for the sake of completeness.

\begin{thm}\label{thm:main}
Fix $n$, and let $\A_n = \{\Id=M^{(0}, \dots, M^{(m)}\}$ be an in-place algorithm computing some linear function in $\R^n$
and let $P,Q\in \R^{n\times n}$ be two matrices.  
For any $t\in[m]$, let $\{i_t,j_t\}$ denote the set of at most two
indices that are affected by the $t$'th gate (if the $t$'th gate is a constant gate, then $i_t=j_t$, otherwise it's a rotation
acting on indices $i_t,j_t$).
Then for any $R \in [\lfloor n/2\rfloor ]$  there exists $t\in [m]$  such that
\begin{equation}\label{eq:thm:main1}
\sqrt{\left \|\row{(M^{(t)}P)}{I_t}\right \|_F^2  \left \|\row{((M^{(t)})^{-T}Q)}{I_t}\right \|_F^2} \geq \frac {R(\Phi_{P,Q}(M^{(m)})-\Phi_{P,Q}(\Id))} {m\log {2R}}\ , \\
\end{equation}
where $I_t = \bigcup_{t'=t}^{t+R-1}\{i_{t'},j_{t'}\}$.
Additionally, if $R=1$ then the $t$'th gate can be assumed to be a rotation.

In particular, if $M^{(m)}=F$ and $m=(n\log n)/b$ for some $b\geq 1$ (``$\A_n$ speeds up FFT by a factor of $b$'') and $P=Q=\Id$, then
\begin{equation}\label{eq:thm:main}
\sqrt{\left \|\row{(M^{(t)})}{I_t}\right \|_F^2  \left \|\row{((M^{(t)})^{-T})}{I_t}\right \|_F^2}  \geq \frac {Rb} {\log {2R}}\ . \\
\end{equation}

\end{thm}



For the main result in this paper in the next section, we will only need the case $R=1$ of the theorem.  It is worthwhile, however,
to state the case of general $R>1$ because it gives rise to a stronger notion of ill-condition than is typically used.
Since this is not the main focus of this work, we omit the details of this discussion.  Henceforth, we will only use the
theorem with $R=1$.

We discuss the implication of the theorem, in case $R=1, P=Q=\Id$.  The theorem implies that
an algorithm  with $m=(n\log n)/b$
must exhibit an intermediate matrix $M^{(t)}$ and a pair of indices $i_t, j_t$ such that the $t$'th gate is a rotation
acting on $i_t, j_t$ and additionally:
$$ \sqrt{\left( \|\row{M^{(t)}}{i_t}\|^2+\|\row{M^{(t)}}{j_t}\|^2\right )\left( \|\row{(M^{(t)})^{-T}}{i_t}\|^2+\|\row{(M^{(t)})^{-T}}{j_t}\|^2\right )} \geq b\ .$$ 
Hence, either
\begin{eqnarray*}
&(i)&\ \ \sqrt{ \|\row{M^{(t)}}{i_t}\|^2+\|\row{M^{(t)}}{j_t}\|^2}\geq \sqrt b 
\mbox{\ \ \ \ \ {\bf -or-} \ \ \ \ \ } \\
&(ii)&\ \  \sqrt{ \|\row{(M^{(t)})^{-T}}{i_t}\|^2+\|\row{(M^{(t)})^{-T}}{j_t}\|^2} \geq \sqrt b\ .
\end{eqnarray*}
\paragraph{Case (i).} We can assume wlog that \begin{equation}\label{eq:overdef}\|\row{M^{(t)}}{i_t}\|^2 \geq b/2\ .\end{equation} 
Let $\xover^T:=\row{M^{(t)}}{i_t}/\|\row{M^{(t)}}{i_t}\|\in \R^n$ ($\xover$ is the normalized $i_t$'th row of $M^{(t)}$, transposed).
%
%
Recall  that the input $x$ is distributed according to the law $\N(0,\Theta(1)\cdot \Id)$.
The $i_t$'th coordinate just before the  $t$'th gate equals $\|\row{M^{(t)}}{i_t}\| x^T \xover$, and is hence distributed
$\N(0,\Theta(\|\row{M^{(t)}}{i_t}\|^2))$.
Using (\ref{eq:overdef}), this  is $\N(0,\Omega(b))$.
If $b=b(n) = \omega(1)$, then by our definition we reach overflow.

Note  that it is possible
as a preprocessing step to replace $x$ with $x-(x^T \xover)\xover$ (eliminating the overflow component), and then to 
reintroduce the  offending component by adding $(x^T \xover)F\xover$ as a postprocessing step.  
In the next section, however, we shall show that, in fact, there must be $\Omega(n)$ pairwise orthonormal directions (in input space)
that  overflow at $\Omega(n)$ different time steps, so such a simple ``hack'' cannot   work.
%
\paragraph{Case (ii).} This scenario,  as the reader guesses, should be called \emph{underflow}.  
In case (ii),  wlog  \begin{equation}\label{tytyty}\|\row{(M^{(t)})^{-T}}{i_t}\|^2 \geq b/2\ .\end{equation}  Now define  $\xunder^T=\row{(M^{(t)})^{-T}}{i_t}/\|\row{(M^{(t)})^{-T}}{i_t}\|\in \R^n$, and consider the orthonormal basis $u_1,\dots u_n\in \R^n$  so  that $u_1=\xunder$.
For any $t'\in [m]$ (and in particular for $t'=t$):
$$ g_1 := \xunder^T x = (\xunder^T (M^{(t')})^{-1})\cdot  (M^{(t')}x)\ .$$
Now notice that the $i_t$'th coordinate of $ (\xunder^T (M^{(t)})^{-1})$ has magnitude at least $\sqrt{b/2}$ by
(\ref{tytyty}) and the construction of $\xunder$.  Also notice that for all $i\neq i_t$, the row $\row{M^{(t)}}{i}$ is orthogonal to $\xover$,
by matrix inverse definition.  This means that coordinate $i\neq i_t$ of $M^{(t)} x$ contains no information about $g_1$.
All the information in $g_1$ is hence contained in $(M^{(t)} x)(i_t)$.  More precisely, $g_1$ is given by 
$ g_1 = ((M^{(t)})^{-T}\xunder)(i_t) \times (M^{(t)} x)(i_t)  - e$,
where $e$ is a random variable independent of $g_1$.  But $|((M^{(t)})^{-T}\xunder)(i_t)| \geq \sqrt{b/2}$,
and $(M^{(t)} x)(i_t)$ is known only up to an additive error of $\eps$, due to our assumptions on quantization in the
numerical architecture.  This means that $g_1$ can only be known
up to an additive error of at least $\eps\sqrt{b/2}$, for \emph{any} value of $e$.
It is important to note that this uncertainty cannot be ``recovered'' later by the algorithm, because at any step the machine state contains all the information about the input
(aside from the input distribution prior).  In other words, any information forgotten at
any step cannot be later recalled (see Figure~\ref{fig:fig} in the Appendix).

Notice that at step $0$, the input vector coordinates $x(1),\dots, x(n)$ are 
represented in individual words, each of which gives rise to an uncertainty interval of width $\eps$.
So merely storing the input in memory in the standard coordinate system implies knowing its location up to
an uncertainty $n$-cube with side $\eps$, and of diameter $\eps\sqrt n$.\footnote{To be precise, we must acknowledge the prior distribution on $x$ which also provides information about its whereabouts.}  An uncertainty interval of size $\eps\sqrt{b/2} = O(\eps\sqrt{\log n})$ in a single direction is therefore relatively benign.
  The next section tells us, however, that the problem is amplified $\Omega(n)$-fold.

\section{Many Independent Ill Conditioned Botlenecks}


\begin{thm}\label{thm:subspaces}
Fix $n$, and let $\A_n = \{\Id=M^{(0)}, \dots, M^{(m)}=F\}$ be an in-place algorithm computing $F$ in time $m=(n\log n)/b$ for some $b\geq 1$.
Then one of the following (i)-(ii) must hold:
\begin{itemize}
\item[(i)] (Severe Overflow) There exists an orthonormal system $v_1,\dots, v_{n'}\in \R^n$ , integers $t_1,\dots, t_{n'}\in [m]$ and $i_1,\dots, i_{n'}\in[n]$ with $n'=\Omega(n)$ such that for all $j\in[n']$, 
\begin{equation}\label{pakapu}\row{M^{(t_j)}}{i_j}P_j = \alpha_j v_j\ \ \mbox{ with}\  \alpha_j=\Omega(\sqrt b) \ ,\end{equation}
where $P_j$ is projection onto the space orthogonal to $v_1,\dots, v_{j-1}$.  
\item[(ii)] (Severe Underflow) There exists an orthonormal system $u_1,\dots, u_{n'}\in\subseteq \R^n$ , integers $t_1,\dots, t_{n'}\in [m]$ and $i_1,\dots, i_{n'}\in[n]$ with $n'=\Omega(n)$ such that for all $j\in[n']$, \begin{equation}\label{eq:severeunder}\row{(M^{(t_j)})^{-T}}{i_j}Q_j = \gamma_j u_j\ \ \mbox{ with}\  \gamma_j=\Omega(\sqrt b) \ ,\end{equation}
where $Q_j$ is projection onto the space orthogonal to $u_1,\dots, u_{j-1}$.
\end{itemize}	

In both cases (i) and (ii), the gates at time $t_1,\dots t_{n'}$ are  rotations, and for all $j\in [n']$ the index $i_j$ is one of the two indices affected by the corresponding rotation.  
Additionally, the set $\{t_1,\dots, t_{n'}\}$ is of cardinality at least $n'/2$.
\end{thm}

The proof heavily relies on Lemma~\ref{lem:Fafterproj} (Section~\ref{sec:lemmas}) and is deferred to Appendix~\ref{sec:proof:subspaces} due to lack of space.  We discuss its  numerical implications, continuing the discussion
following Theorem~\ref{thm:main}.  
 In the severe overflow case, Theorem~\ref{thm:subspaces}
tells us that  there  exists an orthonormal collection $v_1,\dots, v_{n'}$ (with $n'=\Omega(n)$) in  input space, such that each 
$v_i$ behaves like $\xover$ from the previous section. 
This means that, if the speedup factor $b$ is $\omega(1)$, we have overflow  caused by a linear number of independent input components, occurring at
 $\Omega(n)$ different time steps  (by the last sentence in the theorem).
In the extreme
case of speedup $b=\Theta(\log n)$ (linear number of gates), this means that in a constant fraction of time steps overflow
occurs.

For the severe underflow case 
we offer a  geometric interpretation. The theorem  tells us that there exists an orthonormal collection
$u_1,\dots, u_{n'}$ in the input space that is bad in the following sense. For each $j\in [n']$, redefine $g_j=u_j^T x$
to be the input component in  direction $u_j$.  Again, the
variables $g_1,\dots, g_{n'}$ are iid $\N(0,\Theta(1))$.
The first element in the series, $u_1$, can be analyzed as $\xunder$ (from the previous section)
whereby it was argued that before the $t_1$'th  step, the component
$g_1 = u_1^T x$ can only be known to within an interval of width $\Omega(\gamma_1\eps)$, independently of
information from components orthogonal to $u_1$.  We remind the reader that by this we mean that the \emph{width} of the interval is independent, but the
location of the interval  depends smoothly (in fact, linearly) on  information from orthogonal components of $x$ (see Figure~\ref{fig:fig} in the appendix).

As for $u_2,\dots,u_{n'}$:  For each $j\in [n']$, let $z_j := (M^{(t_j)})^{-T}(i_j,:)$.  Therefore $u_1=z_1/\|z_1\|$ and by (\ref{eq:severeunder}), for $j>1$ we 
can  write
$ z_j = \gamma_j u_j + h_j$,
where $h_j \in \span\{u_1,\dots, u_{j-1}\}$.
  Treating $z_j/\|z_j\|$ again as $\xunder$, we
conclude that the component $(z_j/\|z_j\|)^T x $  can only be known to within an interval of size $\Omega(\eps\|z_j\|)$,
given any value of the projection of input $x$ onto the space orthogonal to $z$.



We extend the list of vectors $z_1,\dots, z_{n'}$, orthonormal vectors $u_1,\dots, u_{n'}$, numbers $\gamma_1,\dots, \gamma_{n'}$ and projections $Q_1,\dots, Q_{n'}$ to size $n$ as follows.
Having defined $z_j, u_j, Q_j,\gamma_j$ for some $j\geq n'$, we inductively define $Q_{j+1}$
as projection onto the space orthogonal to $\span\{z_1,\dots, z_{j}\}=\span\{u_1,\dots, u_j\}$ and $z_{j+1}$ to be a standard basis vector such
that $\|Q_{j+1}z_{j+1}\|^2 \geq 1-j/n$. (Such a vector exists because there must exist an index $i_0\in [n]$ such that $\sum_{j'=1}^j u_{j'}(i_0)^2 \leq j/n$, by orthonormality of the collection $u_1,\dots, u_j$;  
Now set $z_{j+1}$ to have a unique $1$ at coordinate $i_0$ and $0$ at all other coordinates.)
We let $u_{j+1}$ be $Q_{j+1}z_{j+1}/\|Q_{j+1}z_{j+1}\|$, that is, a normalized vector pointing
to the component of $z_{j+1}$ that is orthogonal to $\span\{z_1,\dots, z_j\}=\span\{u_1,\dots, u_j\}$.
The number $\gamma_{j+1}$ is  defined as $\|Q_{j+1}z_{j+1}\|$. By construction, $\gamma_{j+1} \geq \sqrt{1-j/n}$.

The above  extends the partial construction arising from the  severe underflow to a full basis,
with the following property:
\begin{prop}\label{prop:ppppp}
For any $j\in [n]$, even given exact knowledge of the exact projection $\tilde x$ of $x$ onto the space orthogonal to $z_j$,
the quantity $x^T(z_j/\|z_j\|)$ upon termination of the algorithm can only be known to within an interval
of the form $[s,s+\eps\|z_j\|]$ where $s$ depends smoothly (in fact, linearly) on $\tilde x$.
\end{prop}
The proposition is simply a repetition of the analysis done for $\xunder$ in the previous section.  For $j>n'$ it is a simple consequence of the
fact that upon initialization of the algorithm with input $x$, each coordinate of $x$  (and in particular $x^Tz_{j}$) is stored
in a single  machine word, while all other machine words store information independent of $x^T z_j$.  Hence the uncertainty of width $\eps\|z_j\|=\eps$.




What do we know about $x$ upon termination of the algorithm?  As stated earlier, any information that
was lost during execution, cannot be later recovered.
Let $\I$ denote the set of possible inputs, given the information the we are left with upon
termination.
Consider the projection $Q_2$ onto the space orthogonal to $u_1=z_1/\|z_1\|$, as
a function defined over $\I$.  Let $\I_2 = Q_2 \I$ denote its image.  The preimage of any point $w\in \I_2$ must contain a line segment
of length at least $\eps \gamma_1$ parallel to $u_1$, due to the uncertainty in $x^Tu_1$.  Hence the volume of $\I$
is at least $\eps\gamma_1$ times the $(n-1)$-volume of $\I_2$.\footnote{We need to be precise about measurability, 
but this is a simple technical point from the fact that the interval endpoint depends smoothly on the projection, as claimed in Proposition~\ref{prop:ppppp}.}
Continuing inductively,  we lower bound the $(n-j+1)$-volume of $\I_j := Q_j \I = Q_j\I_{j-1}$ for $j>2$.  Consider
the projection $Q_{j}$ as a function operating on $\I_{j-1}$, and any point $w$ in the image
$\I_j$.  By definition of $Q_j$, there exists $\hat w \in \I$ such that $Q_j\hat w = w$.  By proposition~\ref{prop:ppppp}, the intersection of the line ${\cal L} = \{\hat w + \eta z_j: \eta\in \R\}$ with $\I$ must contain a segment $\Delta$ of size $\eps \|z_j\|$.  The projection $Q_j\Delta$
of this segment is contained in the line $Q_j{\cal L} = \{w + \eta u_j:\eta \in \R\}$.  The
size of the segment is $\eps\|Q_jz_j\|=\eps \gamma_j$.  This means that the $(n-j+1)$-volume
of $\I_{j+1}$ is at least $\eps\gamma_j$ times the $(n-j)$-volume of $\I_{j+1} = Q_{j+1}\I_j$.


Concluding, we get that the volume of $\I$ is at least $\prod_{j=1}^n \gamma_j$.  
From the construction immediately preceding Proposition~\ref{prop:ppppp}, we get (using the fact that $n'=\Omega(n)$):
$\log \frac{\vol(\I)}{\eps^n} \geq n' \log\sqrt{b/2} + \sum_{j=n'+1}^n \log\sqrt{1-\frac{j-1}{n}} =\Omega(n\log b)$.
This tells us that the volume of uncertainty in the input (and hence, the output) of a $b$-speedup
of FFT in the in-place model is at least $b^{\Omega(n)}$ times the volume of uncertainty incurred simply by
storing the input in memory.

\section{Main Technical Lemma}\label{sec:lemmas}
The following is the most important technical lemma in this work.  Roughly speaking, it tells us that 
application of  operators that are
close to $\Id$  to the rows of $F$ and $F^{-T}$ does not reduce the corresponding potential by much.
Similarly,  assuming that $P,Q$ are PSD with spectral norm at most $1$, applying these transformations to the rows of $\Id$ does not increase the corresponding
potential by much.
\begin{lem}\label{lem:Fafterproj}
Let $P,Q \in \R^{n\times n}$ be two matrices.  
Let $\hat P = \Id-P, \hat Q = \Id-Q$.
Then 
\begin{eqnarray}
\Phi(FP, F^{-T}Q) & \geq& 
n\log n - (\tr\hat P + \tr\hat Q)\log n -  O\left ((\|\hat P\|_F^2+\|\hat Q\|_F^2)\log n\right ) \ .  \label{eq:lem:Fafterproj1}
\end{eqnarray}
If, additionally, $P$ and $Q$ are positive semi-definite contractions, then
\begin{eqnarray}
\Phi_{P,Q}(\Id) &=& \Phi(P,Q)\leq \trace \hat P + \trace \hat Q +  O\left ((\|\hat P\|_F^2+\|\hat Q\|_F^2)\log n\right )  \label{eq:lem:Fafterproj2}\ .
\end{eqnarray}
\end{lem}

The proof, deferred to Appendix~\ref{sec:proof:Fafterproj} for lack of space,
takes advantage of the smoothness of the matrices $F$ and $\Id$ (that is, almost all matrix elements have exactly the same magnitude).
This is the reason we needed to  modify $\Phi$ and work with $\Phi^\C$ for the complex case: If $F$ were  the real
representation of the $n/2$-DFT matrix, then it is not smooth in this sense.  It does hold though that for any $i\in [n]$ and $j\in [n/2]$: $F(i,2j-1)^2+F(i,2j)^2=2/n$,
so the matrix is smooth only in the sense that all pairs of adjacent elements have the same norm (viewed as $\R^2$ vectors).

\section{Future Work}\label{sec:future}

Taking into account bit operation complexity, and using state-of-the-art integer multiplication algorithms
 \cite{DeKSS13,Furer:2007:FIM:1250790.1250800} it can be quite easily shown that both severe overflow and severe
 underflow could be resolved by allowing flexible word size, accommodating either large numbers (in the overflow case)
 or increased accuracy (in the underflow case).  In fact, allowing $O(\log b)$-bit words at the time steps at which overflow (or underflow)
occur, of which there are $\Omega(n)$ many by Theorem~\ref{thm:subspaces}, suffice.  Hence, this work
does not rule out the possibility of (in the extreme case of $b=\Theta(\log n)$) a Fourier transform algorithm in the in-place model using a linear number of gates, in bit operation complexity of $\tilde\Omega(n\log \log n)$, where $\tilde O()$ here
hides $\log\log \log n$ factors arising from fast integer multiplication algorithms.  We conjecture that such
an algorithm does not actually exist, and leave this as the main open problem.


Another problem that was left out in this work is going beyond the in-place model.  In the more general model,
the algorithm works in space $\R^{\ell}$ for $\ell>n$, where the $(\ell-n)$ extra coordinates can be assumed to be
initialized with $0$, and the first $n$ are initialized with the input $x\in \R^n$.  The final matrix $M^{(m)}$ of Fourier transform algorithm $\A_n=\{\Id=M^{(0)},\dots, M^{(m)}\}$ contains $F$ as a sub matrix, so that the output $Fx$ can simply be extracted
from a subset of $n$ coordinates of $M^{(m)} x$, which can be assumed to be the first.    The matrix $M^{(m)}$ (and its inverse-traspose) therefore
contains $(\ell-n)$ extra rows.  The submatrix defined by the extra rows (namely, the last $\ell-n$) and the first $n$ columns were referred to in \cite{Ailon14} as the ``garbage'' part of the computation.  To obtain an $\Omega(n\log n)$ computational
lower bound in the model assumed there,\footnote{In \cite{Ailon14}, the model simply assumed that all matrices $M^{(t)}$ for $t-1\dots m$ have bounded
condition number.  Quantifying the effect of ill condition on numerical stability, overflow and underflow, was not done there. } it was necessary to show that $\Phi_{P,P}(M^{(m)})P=\Omega(n\log n)$, where
$P\in \R^{\ell\times \ell}$ is projection onto the space spanned by the first $n$ standard basis vectors.\footnote{The function  $\Phi_{P,Q}(M)$ was not defined in \cite{Ailon14}, and was only implicitly used.}  To that end,
it was shown that  such a potential lower bound held as long as  spectral norm of the ``garbage'' submatrices was
properly upper bounded. That result, in fact, can be deduced as a simple outcome of Lemma~\ref{lem:Fafterproj} that was developed here.  
What's more interesting is how to generalize Theorem~\ref{thm:subspaces} to the non in-place model, and more importantly
how to analyze the numerical accuracy implications of overflow and underflow to the non in-place model.  Such a generalization
is not trivial and is another immediate open problem following this work.

Another interesting possible avenue is to study the complexity of Fourier transform on input  $x$ for which some prior
knowledge is known.  The best example is when  $Fx$ is assumed sparse, for which much interesting work on the upper bound side
has been recently done by Indyk et al. (see \cite{DBLP:conf/soda/IndykKP14} and references therein).

Many algorithms use the Fourier transform as a subroutine.  In certain cases (fast polynomial multiplication,
fast integer multiplication \cite{DeKSS13,Furer:2007:FIM:1250790.1250800},
fast Johnson-Lindenstrauss transform for dimensionality reduction \cite{DBLP:journals/siamcomp/AilonC09,DBLP:journals/dcg/AilonL09,DBLP:journals/talg/AilonL13,DBLP:journals/siamma/KrahmerW11} and the
related restricted
isometry property (RIP) matrix construction \cite{DBLP:journals/jacm/RudelsonV07,DBLP:journals/corr/abs-1301-0878,DBLP:journals/siamma/KrahmerW11}) the Fourier
transform  subroutine is the algorithm's bottleneck.  Can we use the techniques developed here to derive
lower bounds (or rather, time-accuracy tradeoffs) for those algorithms as well? Moreover, we
can ask how the implications of speeding up the Fourier transform subroutine (as derived in this work)
affect the numerical outcome of these algorithms, assuming they  insist on using Fourier transform as a black box.

\bibliographystyle{plain}
\bibliography{low_bound_fft}

\appendix

\section{Useful Lemmas}
\begin{lem}\label{lem:entrbound}
Let $x,y\in \R^a$ for some integer $a$,  with $\|x\|_2=\|y\|_2=1$.  Then
$-\log a \leq \Phi(x,y) \leq \log a$.
\end{lem}
The proof is a simple done by a simple analysis of the function $\Phi(x,y)$ under the stated constraints using, say, Lagrange multipliers.
\begin{lem}\label{lem:entrchange}
Let $A, B\in \R^{a\times n}$.
Let $U\in R^{a\times a}$ be orthogonal.  Then
$\left | \Phi(A,B) - \Phi(UA, UB) \right | \leq \|A\|_F \|B\|_F \log a$. 
\end{lem}
\begin{proof}
Let $r_i = \|\col{A}{i}\|_2, s_i=\|\col{B}{i}\|_2$.
Then 
\begin{eqnarray*}\Phi(A,B) = \sum_{i=1}^n \Phi(\col{A}{i},\col{B}{i})\ ,& & \Phi(UA, UB) = \sum_{i=1}^n \Phi(U\col{A}{i}, U\col{B}{i})\ ,\end{eqnarray*} and
by the triangle inequality:
$$\left | \Phi(A,B) - \Phi(UA, UB)\right| \leq \sum_{i=1}^n \left | \Phi(\col{A}{i}, \col{B}{i}) - \Phi(U\col{A}{i}, U\col{B}{i})\right |\ .$$  
Fix $i\in [n]$ and let $x,y\in \R^a$ denote $\frac{\col{A}{i}}{r}$, $\frac{\col{B}{i}}{s}$, respectively (note that $\|x\|_2 = \|y\|_2=1$).
\begin{eqnarray}
& &\Phi(\col{A}{i}, \col{B}{i}) - \Phi(U\col{A}{i}, U\col{B}{i})  \\
& &\ \ \ = -\sum_{j=1}^a rs\cdot x(j)y(j) \log | rs\cdot x(j)y(j)| 
 + \sum_{j=1}^a  rs\cdot (Ux)(j)(Uy)(j)\log |rs\cdot (Ux)(j)(Uy)(j)|  \nonumber \\
& & \ \ \ = (rs\log (rs))\left ( -\sum_{j=1}^a x(j) y(j) + \sum_{j=1}^a (Ux)(j)(Uy)(j) \right )  -rs\sum_{j=1}^a x(j)y(j)\log |x(j)y(j)|  \nonumber \\
& & \ \ \ \ \hspace{7cm} + rs\sum_{j=1}^a (Ux)(j)(Uy)(j) \log |(Ux)(j)(Uy)(j) |\ . \nonumber
\end{eqnarray}
By orthogonality of $U$, we have that $\sum_{j=1}^a x(j)y(j)=\sum_{j=1}^a(Ux)(j)(Uy)(j)$.  
Also for the same reason we have $\|Ux\|_2=\|Uy\|_2=1$.  Using Lemma~\ref{lem:entrbound}, we conclude
$\left |\Phi(\col{A}{i}, \col{B}{i}) - \Phi(U\col{A}{i}, U\col{B}{i}) \right |\leq 2rs\cdot \log a$.  
Summing up over $i$ and applying Cauchy-Schwarz we conclude the result.
\end{proof}

\begin{lem}\label{lem:entrchange2}
Let $A,B \in \R^{a\times n}$, and let $D\in \R^{a\times a}$ be some nonsingular matrix.
Then
\begin{equation}\label{eq:entrchange2}
\left |\Phi(A,B) - \Phi(DA, D^{-T}B)\right | \leq \left(\|A\|_F\|B\|_F + \|DA\|_F\|D^{-T}B\|_F\right )\log a
\end{equation}
\end{lem}
\begin{proof}
Let $U,V\in \R^a$ be orthogonal and $\Sigma\in \R^a$ diagonal (and nonsingular) so that $D=U\Sigma V$.  (Such a composition exists
by standard SVD theory.)
\begin{eqnarray}\label{eq:entrchange2}
\left | \Phi(A,B) - \Phi(DA, D^{-T}B)\right |&\leq& \left |\Phi(A,B) - \Phi(VA, VB)\right | \nonumber \\
& & \ \ \ +  \left | \Phi(VA, VB) - \Phi(\Sigma VA, \Sigma^{-1} VB) \right | \nonumber \\
& & \ \ \ +  \left | \Phi(\Sigma VA, \Sigma^{-1} VB) - \Phi(U\Sigma V A, U\Sigma^{-1} VB) \right | \nonumber \\
&\leq& \|A\|_F\|B\|_F\log a + 0 + \|\Sigma V A\|_F\|\Sigma^{-1} V B\|_F\log a \nonumber \\
&=&\left ( \|A\|_F\|B\|_F +  \| D A\|_F\|D^{-T} B\|_F\right )\log a\ , \nonumber \\
\end{eqnarray}
as required.  (We used Lemma~\ref{lem:entrchange} twice in the second  inequality, and the orthogonality of $U$ for the last derivation.  The reason the middle term in the RHS of the first inequality is null is by properties of $\Phi$ that are trivial to check.)
\end{proof}

\section{Proof of Theorem~\ref{thm:main}}\label{sec:proof:thm:main}

We directly prove the less general (\ref{eq:thm:main}).  The more general bound (\ref{eq:thm:main1}) is shown similarly, but with more notation.
Fix $R\in [\lfloor n/2\rfloor ]$.  Let $m'$ be the smallest integer divisible by $R$ satisfying
$m'\geq m$.  If $m'>m$, then ``pad'' the algorithm $\A_n$ by defining $M^{(m+1)}\dots M^{(m')} = M^{(m)}=F$.
By the triangle inequality,
\begin{eqnarray} \label{eq:triangleeq}
\left | \Phi(M^{(m')})  - \Phi(M^{(0)}) \right | &\leq& \sum_{j=1}^{m'/R}\left | \Phi(M^{(jR)}) - \Phi(M^{((j-1)R)})\right |\ .
\end{eqnarray}
Now note that for each $j\in m'/R$, the matrix $M^{(jR^*)}$ is obtained from $M^{((j-1)R^*)}$ by applying
a nonsingular operation acting on the left, affecting at most $2R$ rows.  Denote the set of indices of the corresponding set of affected rows  by $I_j$.  (If the cardinality of
$I_j$ is less than $2R$, then pad it with an arbitrary set of indices.)
 Using Lemma~\ref{lem:entrchange2}, this implies that for all $j\in[m'/R]$,

\begin{eqnarray}
& & \left | \Phi(M^{(jR)}) - \Phi(M^{((j-1)R)})\right |  \label{eq:uselemma} \\
& & \ \ \ \leq
\left ( \| \row{M^{(jR)}}{I_j} \|_F \|\row{(M^{(jR)})^{-T}}{I_j}  \|_F + \|\row{M^{((j-1)R)}}{I_j}  \|_F \| \row{(M^{((j-1)R^*)})^{-T}}{I_j}  \|_F \right )\log 2R \ . \nonumber
\end{eqnarray}
Combining (\ref{eq:uselemma}) with (\ref{eq:triangleeq}), we get
\begin{eqnarray*}
\left | \Phi(M^{(L)})   -\Phi(M^{(0)}) \right |
 \leq 2\left (  \sum_{j=0}^{m'/R} \| \row{M^{(jR)}}{I_j} \|_F \|\row{(M^{(jR)})^{-T}}{I_j}  \|_F\right )\log 2R\ . \nonumber
\end{eqnarray*}

For any matrix $A\in \R^{n\times n}$ and any subset $I\subseteq[n]$, we have $\|\row{A}{I}\|_F^2\leq \sum_{i=1}^{|I|}\sigma_i^2(A)$ (this can be seen e.g. using the SVD theorem).  Therefore,
\begin{eqnarray*}
\left | \Phi(M^{(m')})   -\Phi(M^{(0)}) \right |
 &\leq& 2\left (  \sum_{j=0}^{m'/R} \sqrt{\sum_{i=1}^{2R} \sigma_i^2(M^{(jR)})\sum_{i'=1}^{2R} \sigma_{i'}^2(M^{(jR)})^{-T}) } \right )\log 2R  \nonumber \\
&=& 2\left (  \sum_{j=0}^{m'/R} \sqrt{\sum_{i=1}^{2R} \sigma_i^2(M^{(jR)})\sum_{i'=1}^{2R} \sigma_{n-i'+1}^{-2}(M^{(jR)}) } \right )\log 2R\ .
\end{eqnarray*}

But  $\Phi(M^{(m')}) =  \Phi(F) = n\log n$ and $\Phi(M^{(0)})=0$, hence, there must exists $j\in [m'/R]$ with
\begin{eqnarray*}
 \sqrt{\sum_{i=1}^{2R} \sigma_i^2(M^{(jR)})\sum_{i'=1}^{2R} \sigma_{n-i'+1}^{-2}(M^{(jR)}) } &\geq& \frac{2Rn\log n}{m'\log 2R} \geq \frac{2Rn\log n}{((n\log n)/b+R)\log 2R} \\
&\geq& \frac{2Rn\log n}{2((n\log n)/b)\log 2R}  = \frac{Rb}{\log 2R}\ .
\end{eqnarray*}

\section{Proof of Lemma~\ref{lem:Fafterproj}}\label{sec:proof:Fafterproj}

\def\Ismall{I_{\operatorname{small}}}
\def\Ibig{I_{\operatorname{big}}}
\def\smll{{\operatorname{small}}}
\def\bg{{\operatorname{big}}}
We start by proving (\ref{eq:lem:Fafterproj1}).
For brevity, we denote $F(i,j)$ by $f_{i,j}\in \{1/\sqrt n, -1/\sqrt n\}$, $(F\hat P)(i,j) = \epsilon_{i,j}$, $(F\hat Q)(i,j) = \delta_{i,j}$.  Therefore, $(FP)(i,j) = f_{i,j} - \epsilon_{i,j}$ and $(FQ)(i,j) = f_{i,j} - \delta_{i,j}$.
 Let $\|\epsilon\|_F = \sqrt{\sum_{i=1}^n\sum_{j=1}^n \epsilon_{i,j}^2}$ and $\|\delta\|_F =\sqrt{\sum_{i=1}^n\sum_{j=1}^n \delta^2_{i,j}}$.  By orthogonality of $F$, we have that 
 \begin{eqnarray}\|\epsilon\|_F=\|{\hat P}\|_F =: \alpha &\ \ \ \  & \|\delta\|_F=\|{\hat Q}\|_F =: \beta \label{eq:boundepsdelta} \\
\sum_{i=1}^n\sum_{j=1}^n f_{i,j}\epsilon_{i,j} = \trace F^T F\hat P = \trace {\hat P} &\ \ \ \  &
\sum_{i=1}^n\sum_{j=1}^n f_{i,j}\delta_{i,j} = \trace F^T F\hat Q = \trace {\hat Q}\ . \label{eq:fromprojection} %
\end{eqnarray}
\noindent
Let $I_1,I_2,I_3,I_4\subseteq [n]\times [n]$ be defined as
\begin{eqnarray*}
I_1 := \{(i,j): \epsilon_{i,j}^2 < 1/(2n) \mbox{ and }\delta_{i,j}^2 < 1/(2n)\}  & &
I_2 := \{(i,j): \epsilon_{i,j}^2 < 1/(2n) \mbox{ and }\delta_{i,j}^2 \geq 1/(2n)\} \\
I_3 := \{(i,j): \epsilon_{i,j}^2 \geq 1/(2n) \mbox{ and }\delta_{i,j}^2 < 1/(2n)\}  & &
I_4 := \{(i,j): \epsilon_{i,j}^2 \geq 1/(2n) \mbox{ and }\delta_{i,j}^2 \geq 1/(2n)\}\ . \\
\end{eqnarray*}
Now, we write $\Phi(FP, FQ)$ as $\Phi_1 +\Phi_2+\Phi_3+\Phi_4$, where $ \forall h=1,2,3,4$:
$$\Phi_h := -\sum_{(i,j)\in I_h} (f_{i,j} - \epsilon_{i,j})(f_{i,j}-\delta_{i,j})\log | (f_{i,j} - \epsilon_{i,j})(f_{i,j}-\delta_{i,j})|\ .$$

\noindent We start by bounding $\Phi_4$.  For any $(i,j)\in I_4$, 
\begin{eqnarray}|f_{i,j} - \epsilon_{i,j}| \leq 3|\epsilon_{i,j}|  \label{eq:esteps}, \ \ & & 
|f_{i,j} - \delta_{i,j}| \leq 3|\delta_{i,j}|  \label{eq:estdelta}\ .
\end{eqnarray}
Write $I_{4.1} \cup  I_{4.2}$, where $I_{4.1} = \{(i,j): 9|\epsilon_{i,j}\delta_{i,j}| \leq 1/e\}$ and $I_{4.2} = I_4\setminus I_{4.1}$.
Accordingly, for $r=1,2$:  
$\Phi_{4.r} := -\sum_{(i,j)\in I_{4.r}}  |(f_{i,j}-\epsilon_{i,j})(f_{i,j}-\delta_{i,j})|\log 
|(f_{i,j}-\epsilon_{i,j})(f_{i,j}-\delta_{i,j})|$.
Using (\ref{eq:esteps})-(\ref{eq:estdelta})  and the monotonicity (increasing) of $-x\log x$ in the range $x\in[0,1/e]$,
we conclude
\begin{eqnarray}
|\Phi_{4.1}| &\leq& -\sum_{(i,j)\in I_4}  9|\epsilon_{i,j}\delta_{i,j}|\log  9|\epsilon_{i,j}\delta_{i,j}| 
=  -\sum_{(i,j)\in I_4}  9\|\epsilon\|\|\delta\|\left |\frac{\epsilon_{i,j}\delta_{i,j}}{\|\epsilon\|\|\delta\|}\right |\log  9\|\epsilon\|\|\delta\|\left |\frac{\epsilon_{i,j}\delta_{i,j}}{\|\epsilon\|\|\delta\|}\right |  \nonumber \\
&=&  -\sum_{(i,j)\in I_4}  9\|\epsilon\|\|\delta\|\left |\frac{\epsilon_{i,j}\delta_{i,j}}{\|\epsilon\|\|\delta\|}\right |\log  9\|\epsilon\|\|\delta\| -\sum_{(i,j)\in I_4}  9\|\epsilon\|\|\delta\|\left |\frac{\epsilon_{i,j}\delta_{i,j}}{\|\epsilon\|\|\delta\|}\right |\log  \left |\frac{\epsilon_{i,j}\delta_{i,j}}{\|\epsilon\|\|\delta\|}\right |   \nonumber \\
&\leq& 9\alpha\beta \log 9\alpha\beta + 18 \alpha\beta \log n 
\leq 27 \alpha\beta \log n + 9\alpha\beta\log 9\ . \nonumber
\end{eqnarray}
where the second inequality used (\ref{eq:boundepsdelta}), Lemma~\ref{lem:entrbound} and Cauchy-Schwarz.
To bound $|\Phi_{4.2}|$, note that by Cauchy-Schwarz $\sum_{(i,j)\in I_{4.2}} |\epsilon_{i,j}\delta_{i,j}|\leq \alpha\beta$, and hence $|I_{4.2}| \leq 9e\alpha\beta \leq 27\alpha\beta$.  This implies that $|\Phi_{4.2}| \leq 27\alpha\beta$.  Combining, we conclude

\begin{equation}\label{boundbig}
|\Phi_4| \leq |\Phi_{4.1}| + |\Phi_{4.2}| \leq 27\alpha\beta\log n + 63\alpha\beta \leq (\alpha^2+\beta^2)(63+27\log n)\ .
\end{equation}

We now bound $|\Phi_3|$.  For all $(i,j)\in I_3$, (\ref{eq:esteps}) holds.  Again we need to consider two cases, by defining
$I_{3.1} := \{(i,j)\in I_3: 3|\epsilon_{i,j}| \leq 1/e\}$ and $I_{3.2} = I_3 \setminus I_{3.2}$ and, as above,
$\Phi_{3.1}$ and $\Phi_{3.2}$ in an obvious way.    Then,
\begin{eqnarray}
|\Phi_{3.1}| &\leq& -3\sum_{(i,j)\in I_{3.1}} |\epsilon_{i,j}||f_{i,j} - \delta_{i,j}|\log (3|\epsilon_{i,j}||f_{i,j}-\delta_{i,j}|)  \nonumber \\
&=&   -3\sum_{(i,j)\in I_{3.1}} |\epsilon_{i,j}||f_{i,j} - \delta_{i,j}|\log 3|\epsilon_{i,j}|  |f_{i,j}| -3\sum_{(i,j)\in I_{3.1}} |\epsilon_{i,j}||f_{i,j} - \delta_{i,j}|\log |1-\delta_{i,j}/f_{i,j}| \nonumber\\
&\leq &  3 \sum_{(i,j)\in I_{3.1}} |\epsilon_{i,j}||f_{i,j} - \delta_{i,j}| \log n  +  3\sum_{(i,j)\in I_{3.1}} |\epsilon_{i,j}||f_{i,j} - \delta_{i,j}| |\delta_{i,j}/f_{i,j}| \nonumber \\
&\leq& 3 \sum_{(i,j)\in I_{3.1}} |\epsilon_{i,j}|\frac 1{2\sqrt n} \log n + 3\sum_{(i,j)\in I_{3.1}} |\epsilon_{i,j}|/(2\sqrt {2n})\ . \label{eq:psi3.1}
\end{eqnarray}
By Cauchy-Schwarz,
 $\sum_{(i,j)\in I_{3.1}} |\epsilon_{i,j}| \leq \sqrt{|I_{3.1}|}\sqrt{\sum_{(i,j)\in I_{3.1}} \epsilon_{i,j}^2}  \leq \sqrt{|I_3|}\|\epsilon\| =  \sqrt{|I_3|}\alpha \ .$
 But by (\ref{eq:boundepsdelta}) and definition of $I_3$, we have $|I_3| \leq 2\alpha^2n$.  Combining with (\ref{eq:psi3.1}), we get
 \begin{equation}\label{eq:boundPsi3.1}
 |\Phi_{3.1}| \leq 3 \alpha^2 \log n + 3\alpha^2\ .
 \end{equation}

To bound $|\Phi_{3.2}|$, note that by (\ref{eq:boundepsdelta}) and by definition of $I_{3.2}$, $|I_{3.2}| \leq 9e^2\alpha^2 \leq 81\alpha^2$.  But clearly $|\Phi_{3.2}| \leq |I_{3.2}|$, hence  $|\Phi_{3.2}| \leq 81\alpha^2$.
Combining with (\ref{eq:boundPsi3.1}), we conclude
\begin{equation}\label{eq:boundPsi3}|\Phi_3| \leq 3\alpha^2\log n + 84\alpha^2\ .
\end{equation}
By symmetry, we also have:
\begin{equation}\label{eq:boundPsi2}|\Phi_2| \leq 3\beta^2\log n + 84\beta^2\ .
\end{equation}

\noindent
We now turn to approximate $\Phi_1$.
\begin{eqnarray*}
\Phi_1 &=& 
-\sum_{(i,j)\in I_1} f_{i,j}^2\log f_{i,j}^2 
-\underbrace{\sum_{(i,j)\in I_1} f_{i,j}^2\log ((1-\epsilon_{i,j}/f_{i,j})(1-\delta_{i,j}/f_{i,j}))}_{A} \\
& & +\underbrace{\sum_{(i,j)\in I_1} f_{i,j}(\epsilon_{i,j}+\delta_{i,j})\log ( (f_{i,j} - \epsilon_{i,j})(f_{i,j}-\delta_{i,j}))}_B 
 -\underbrace{\sum_{(i,j)\in I_1} \epsilon_{i,j}\delta_{i,j}\log ((f_{i,j} - \epsilon_{i,j})(f_{i,j}-\delta_{i,j})}_C\ .
\end{eqnarray*}
Subtracting $n\log n = -\sum_{i=1}^n\sum_{j=1}^n f_{i,j}^2\log f_{i,j}^2$ from both sides gives:
\begin{eqnarray*}
\Phi_1 - n\log &=& \underbrace{\sum_{(i,j)\not\in I_1} f_{i,j}^2\log f_{i,j}^2}_D -A+B-C\ .
\end{eqnarray*}
By (\ref{eq:boundepsdelta}), the number of pairs $(i,j)\in [n]\times [n]$ for which $\epsilon_{i,j}^2 \geq 1/(2n)$ is at most $2\alpha^2 n$.  Similarly, the number of pairs $(i,j)$ for which $\delta_{i,j}^2 \geq 1/(2n)$ is at most $2\beta^2 n$.
Hence, 
$|([n]\times[n])\setminus I_1|\leq 2(\alpha^2+\beta^2)n$.
  Therefore, $|D| \leq 2(\alpha^2+\beta^2)n(1/n)\log n = 2(\alpha^2+\beta^2)\log n$.
For $A$, we just need to notice that $A\leq 0$.
As for $B$,
\begin{eqnarray*}
B &=& \sum_{(i,j)\in I_1} f_{i,j}(\epsilon_{i,j}+\delta_{i,j})\log f_{i,j}^2 + \sum_{(i,j)\in I_1} f_{i,j}(\epsilon_{i,j}+\delta_{i,j})\log (1-\epsilon_{i,j}/f_{i,j})(1-\delta_{i,j}/f_{i,j}) \\
&=& -\sum_{(i,j)\in I_1} f_{i,j}(\epsilon_{i,j}+\delta_{i,j})\log n 
+ \sum_{(i,j)\in I_1} f_{i,j}(\epsilon_{i,j}+\delta_{i,j})\log (1-\epsilon_{i,j}/f_{i,j})(1-\delta_{i,j}/f_{i,j})
\end{eqnarray*}
Adding $\sum_{i=1}^n\sum_{j=1}^n f_{i,j}(\epsilon_{i,j}+\delta_{i,j})\log n = (\tr{\hat P}+\tr{\hat Q})\log n$ (see (\ref{eq:fromprojection})) to both sides,
taking absolute value on both sides, using the triangle inequality and the estimate $\log(1-\alpha)\leq |\alpha|$ for all $|\alpha|\leq 1/\sqrt 2$:
\begin{eqnarray*}
\left |B + (\tr{\hat P}+\tr{\hat Q})\log n\right | &\leq& \sum_{(i,j)\not\in I_1} |f_{i,j}|(|\epsilon_{i,j}|+|\delta_{i,j}|)\log n 
+\sum_{(i,j)\in I_1} (\epsilon_{i,j}^2 + \delta_{i,j}^2 +  2|\epsilon_{i,j}||\delta_{i,j}|) \\
&\leq&  \sum_{(i,j)\not\in I_1}\frac 1{\sqrt n} ((|\epsilon_{i,j}|+|\delta_{i,j}|)\log n )
+3(\alpha^2+\beta^2)  \\
&\leq& 2(\alpha^2+\beta^2)\log n + 3(\alpha^2+\beta^2)\ .
\end{eqnarray*}
Where the second inequality used Cauchy-Schwarz, and the third used Cauchy-Schwarz to obtain 
$\sum_{(i,j)\neq I_1}|\eps_{i,j}| \leq \sqrt{n^2-|I_1|}\alpha$ together with the estimate
$(n^2-|I_1|) \leq 2n\alpha^2$ (from the definition of $I_1$ and $\alpha$), and a similar
step for bounding $\sum_{(i,j)\neq I_1}|\delta_{i,j}|$.
To bound $|C|$, note that $|\log((f_{i,j}-\epsilon_{i,j})(f_{i,j}-\delta_{i,j}))| \leq 4\log n$ for all $(i,j)\in I_1$.  Hence, using Cauchy-Schwarz,
$ |C| \leq 4\sum_{(i,j)\in I_1} |\epsilon_{i,j}||\delta_{i,j}|\log n \leq 4\alpha\beta\log n \leq 4(\alpha^2+\beta^2)\log n$.
Combining our upper bound for $A$ and  estimates for $B, |C|$ and $|D|$, we conclude:
\begin{equation}\label{eq:boundPsi1}
\Phi_1 \geq n\log n- (\trace \hat P + \trace \hat Q)\log n - (\alpha^2+\beta^2)(6 +8\log n)\ .
 \end{equation}

\noindent
Finally, by combining (\ref{boundbig}), (\ref{eq:boundPsi3}), (\ref{eq:boundPsi2}), (\ref{eq:boundPsi1}), we conclude 
\begin{equation*}
\Phi(FP, FQ) \geq n\log n- (\trace \hat P + \trace \hat Q)\log n - (\alpha^2+\beta^2)(147 +30\log n)\ .
 \end{equation*}
This concludes the proof of (\ref{eq:lem:Fafterproj1}).  

\def\Phidiag{\Phi_{\operatorname{diag}}}
\def\Phioff{\Phi_{\operatorname{off}}}
We now prove (\ref{eq:lem:Fafterproj2}), whence we assume that $P,Q$ are PSD contractions (as are
$\hat P, \hat Q$). We decompose $\Phi_{P,Q}(\Id,\Id)$ as two sums, as follows:
\begin{eqnarray} 
\Phidiag &=& -\sum_{i=1}^n P(i,i)Q(i,i)\log|P(i,i)Q(i,i)| \nonumber \\
\Phioff &=& - \sum_{i\neq j} P(i,j)Q(i,j)\log|P(i,j)Q(i,j)|\ . \nonumber 
\end{eqnarray}
We start by bounding $\Phidiag$. Define $\eta$ as:
$$ \eta := \sum_{i=1}^n P(i,i)Q(i,i) = \sum_{i=1}^n(1-\hat P(i,i))(1-\hat Q(i,i)) \ .$$
Notice that $\eta \geq n - \tr \hat P - \tr \hat Q$, and that $0\leq P(i,i)Q(i,i)\leq 1$ for all $i\in [n]$ by the contraction property.  Using standard tools (e.g. Lagrange multipliers), it can be shown
that the function $-\sum y_i \log y_i$ under the constraints $0\leq y_i\leq 1, \sum y_i=\eta$ obtains
its maximum when $\forall i\in [n]: y_i=1/\eta$, at which case its value is $-\eta \log(\eta/n)$.
Hence,
\begin{eqnarray} \Phidiag \leq -\eta \log(\eta/n) &\leq& -(n-\trace \hat P - \trace \hat Q)\log(1 - \trace \hat P/n - \trace \hat Q/n) \nonumber \\
&\leq & -n\log(1 - \trace \hat P/n - \trace \hat Q/n) \nonumber \\
&\leq& n(\trace \hat P/n + \trace \hat Q/n + (\trace \hat P + \trace \hat Q)^2/n^2) \nonumber \\
&=& \trace \hat P  + \trace \hat Q + (\trace \hat P + \trace \hat Q)^2/n \nonumber \\
&\leq& \trace \hat P  + \trace \hat Q  + \|\hat P + \hat Q\|_F^2 \nonumber \\
&\leq& \trace \hat P + \trace \hat Q + \|\hat P\|_F^2 + \|\hat Q\|_F^2\ .
\end{eqnarray}
(we used twice the assumption that $n-\tr \hat P - \tr \hat Q \geq n/e$, for otherwise $\|\hat P\|_F^2+\|\hat Q\|_F^2 = \Omega(n)$ and (\ref{eq:lem:Fafterproj2}) is trivial.)

We now turn to bound $\Phioff$.  Let $\mu_P = \sum_{i\neq j} P(i,j)^2, \mu_Q = \sum_{i\neq j} Q(i,j)^2$.
If $\mu_P\mu_Q \leq 1$ then without loss of generality  $\mu_P\leq 1$, implying $\|\hat P\|_F = \Omega(n)$, and therefore (\ref{eq:lem:Fafterproj2}) is trivial.  Hence we assume $\mu_P\mu_Q\geq 1$.
 Using Lemma~\ref{lem:entrbound} (define $x$ to be the $n(n-1)$-dimensional vector with $x_{ij}=P(i,j)/\sqrt{\mu_P}$ for $i\neq j$, and similarly define $y$ using $Q$):
\begin{eqnarray*}
\Phioff &\leq & \sqrt{\mu_P \mu_Q} \log n(n-1) - \sum_{i\neq j}P(i,j)Q(i,j) \log \sqrt{\mu_P \mu_Q} \nonumber \\
&\leq& 2\sqrt{\mu_P \mu_Q} \log n \leq (\mu_P+\mu_Q)\log n\ ,
\end{eqnarray*}
where we used the AMGM inequality in the last step. 
But now notice that $\mu_P \leq \|\hat P\|_F^2, \mu_Q \leq \|\hat Q\|_F^2$.  Combining this with
our bound of $\Phidiag$ completes the proof.

\section{Further Discussion on Numerical Architecture and Overflow Definition }\label{sec:discussion}
\begin{enumerate}
\item
Our definition of overflow is counterintuitive, because we are used to thinking about overflow as an offending
machine state at a particular step of  the algorithm execution for a particular input, while our definition is 
stochastic.   The reason we use this definition is from the combination of (a) our desire to work with a spherically symmetric
input and (b) avoiding measuring complexity in the granularity of logical bit operations, stemming from the varying
word length typically arising even in the standard FFT benchmark.
It is possible to somewhat practically justify the stochastic
definition of overflow by thinking of running FFT on a large number $L$  of iid inputs, which can be thought of as being
stacked as columns of an input matrix $X$.  The larger $L$ is, the more concentrated the total number of bits required
to encode each \emph{row} of the matrix around $\Theta(L)$ will be.  For $L$ polynomial in $n$, the probability
of requiring more than $\Theta(L)$ bits per row becomes exponentially small at any step of standard FFT.
\item
We are  ignoring the fact that the machine state after $t$ steps on input $x$ in any computer is not only a result of quantizing the vector $M^{(t)} x$.  Rather, errors are accumulated from the effects of quantizing at earlier steps.  Taking accumulated errors into account should affect the numerical accuracy of both the benchmark and of any speed-up,  but 
quantifying this effect seems extremely difficult.    We take an information theoretical approach by saying that
the machine state after $t$ steps contains \emph{at most} the information in the quantization of the coordinates
of $M^{(t)} x$, and whatever information  lost  (due to this quantization)  cannot be later recovered   because the machine state
encodes everything that is known about the input (and output) at any given step.
\end{enumerate}


\section{Proof of Theorem~\ref{thm:subspaces}}\label{sec:proof:subspaces}
Let $V=\{v_1,\dots, v_k\}\subseteq \R^n$ be an orthonormal set satisfying the properties described in case (i)
of the theorem.
Similarly, let $U=\{u_1,\dots, u_{\ell}\}\subseteq \R^n$ be an orthonormal set satisfying the properties described in case (ii)
of the theorem. 
We show that as long as $k+\ell$ is at most $O(n)$, then we can extend one of the two sets by one element.

Let $P$ denote the projection onto $(\span V)^\perp$ and $Q$ the projection onto $(\span U)^\perp$.
Using Lemma~\ref{lem:Fafterproj}, we have that
\begin{equation} \Phi_{P, Q}(F) - \Phi_{P,Q}(\Id) \geq  n\log n - (\trace \hat P + \trace \hat Q)(1+\log n) -C(\|\hat P\|_F^2+\|\hat Q\|_F^2)\log n\ , \nonumber
\end{equation}
for some global $C>0$, where $\hat P = \Id - P, \hat Q = \Id-Q$ (the orthogonal projections).  By known properties of 
projection matrices, $\trace \hat P = \|\hat P\|_F^2 = k, \trace \hat Q = \|\hat Q\|_F^2 = \ell$.
Therefore,
\begin{equation} \Phi_{P, Q}(F) - \Phi_{P,Q}(\Id) \geq  n\log n - C'(k+\ell)\log n , \nonumber
\end{equation}
for some global $C'>0$.
This implies, using Theorem~\ref{thm:main} (with $R=1$) that for some $t\in\{0,\dots, m\}$ and $i\in [n]$:
\begin{eqnarray}
\|\row{M^{(t)}P}{i}\|\cdot\|\row{(M^{(t)})^{-T}Q}{i}\| &\geq& \frac{  n\log n - C'(k+\ell)\log n}{m}  \nonumber \\
&=& \frac{ (n\log n - C'(k+\ell)\log n)b }{n\log n}\ . \nonumber
\end{eqnarray}

\noindent Hence, as long as $k+\ell \leq n/(2C')$:
\begin{eqnarray}
\|\row{M^{(t)}P}{i}\|\cdot\|\row{(M^{(t)})^{-T}Q}{i}\| &\geq& b/2\ .
\end{eqnarray}

This implies that either $\|\row{M^{(t)}P_\perp}{i}\|\geq \sqrt{b/2}$ or that $\|\row{(M^{(t)})^{-T}Q_\perp}{i}\|\geq \sqrt{b/2}$.  In the former case we can extend the set $V$ by adding $v_{k+1} = \row{M^{(t)}P}{i}/\|\row{M^{(t)}P}{i}\|$, which is orthogonal to $v_1,\dots,v_k$ by construction.  In the latter case we can
extend the set $U$ by adding $u_{k+1}=\row{(M^{(t)})^{-T}Q}{i}/\|\row{(M^{(t)})^{-T}Q}{i}\|$
which is again, orthogonal to $u_1,\dots, u_\ell$ by construction. 

This process of augmenting $V$ and $U$ can continue until $k+\ell \geq n/2C'$, which implies that either $k\geq n/4C'$ 
(establishing extreme overflow of the theorem) or $\ell\geq n/4C'$ (establishing extreme underflow).

Assume $k\geq n/4C'$ and let $n'=k$. 
For $j\in [n']$ let $t_j$ denote the time step of the overflow corresponding to direction $v_j \in \R^n$ and let
$i_j$ denote the coordinate at which the overflow occurs.  It is clear from the construction that $i_j$ is one of the
at most two coordinates affected by the $t_j$'th step.  We show that there exist no $1\leq j < j'\leq n'$ such that $(t_j,i_j)=(t_{ j'}, i_{ j'})$.  Indeed, note that $\row{M^{(t_j)}P_{ j'}}{i_j}$ must be null  by
construction, contradicting the fact that $\|\row{M^{(t_{j})}P_{ j'}}{i_j}\|=\|\row{M^{(t_{ j'})}P_{\hat j'}}{i_{ j'}}\| =\Omega(\sqrt b)$.  The conclusion is that for any $t\in [m]$ there can
be at most two indices $j,j'\in[n']$ such that $t_j=t_{j'}=t$, and therefore the cardinality of the
set $\{t_1,\dots, t_{n'}\}$ is at least $n/2$.
  A similar argument is done for the extreme underflow case, concluding the proof.

\newpage
\section{An Illustrative Figure}
\begin{figure}[h!tbp]
\begin{center}
\psfrag{3}{\scalebox{0.7}{$g_1$}}
\psfrag{2}{\scalebox{0.7}{$e$}}
\psfrag{1}{\scalebox{0.7}{$\geq \eps\sqrt{b/2}$}}
\scalebox{1.5}{
\psfragfig[width=1.0in,keepaspectratio]{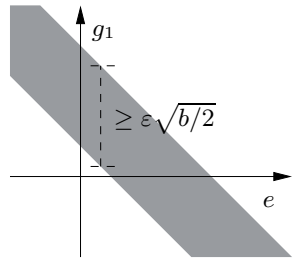}
}
\end{center}
\caption{The uncertainty interval in $g_1$ as a function of $e$, given the representation of $M^{(t)}x(i_t)$ in a computer word.  The random variable $e$ contains only information from components of the input that are orthogonal  to $\xunder$.  
\label{fig:fig}}
\end{figure}

\end{document}